\documentclass[runningheads]{llncs}
\usepackage{color}
\usepackage{graphicx}
\usepackage{subfigure}
\usepackage{braket}
\usepackage[final]{fixme}
\usepackage{amsmath,amsfonts,amssymb}
\usepackage{algorithm}
\usepackage[noend]{algpseudocode}
\usepackage{appendix}
\usepackage[normalem]{ulem}

\newcommand{\half}{\frac 12}

\newcommand{\norm}[2]{\| #1 \|_{ #2 }}
\newcommand{\vnorm}[1]{\norm{ #1 }{V}}

\newcommand{\ltwonorm}[1]{\norm{ #1 }{2}}
\newcommand{\diff}[1]{\text{d} #1}

\newcommand{\Rd}{\mathbb{R}^{d}}
\newcommand{\RdM}{\mathbb{R}^{d\times M}}
\newcommand{\RdK}{\mathbb{R}^{d\times K}}
\newcommand{\nuinf}{\nu_\text{inf}}

\begin{document}
\title{Selective metamorphosis for growth modelling with applications to landmarks}
\titlerunning{Selective metamorphosis for growth modelling}

\author{Andreas Bock \and  Alexis Arnaudon \and Colin Cotter}

\authorrunning{Bock et al.}
%
\institute{Imperial College London}
%
\maketitle              

\begin{abstract}
We present a framework for shape matching in computational anatomy allowing
users control of the degree to which the matching is diffeomorphic. This control
is given as a function defined over the image and parameterises the template
deformation. By modelling localised template deformation we have a mathematical
description of growth only in specified parts of an image. The location can
either be specified from prior knowledge of the growth location or learned from
data. For simplicity, we consider landmark matching and infer the distribution of
a finite dimensional parameterisation of the control via Markov chain Monte
Carlo.  Preliminary numerical results are shown and future paths of
investigation are laid out. Well-posedness of this new problem is studied
together with an analysis of the associated geodesic equations. 
\end{abstract}

\section{Introduction}

In computational anatomy
\cite{grenander1994representations,grenander1998computational} one of the most
fundamental problems is to continuously deform an image or shape into another
and thereby obtain a natural notion of distance between them as the energy
required for such a deformation. Common methods to compute image deformations
are based on diffeomorphic deformations which assume that the images are
continuously deformed into one another with the additional property that the
inverse deformation is also continuous.  This is a strong requirement for images
which implies that the 'mass' of any part of the image is conserved: we cannot
create or close 'holes'.  This is also a crucial property in fluid mechanics and
in fact the theory of diffeomorphic matching carrying the moniker \emph{Large
Deformation Diffeomorphic Metric Mapping} (LDDMM)
\cite{trouve1998diffeomorphisms,beg2005computing} has been inspired by fluid
mechanics. Indeed, Arnold \cite{arnold1966geometrie} made the central
observation that the geodesic equations for the diffeomorphism group induced by
divergence-free vector fields corresponded to that of incompressible flows. If a
strictly diffeomorphic matching is not possible or necessary, an extension of
LDDMM called metamorphosis \cite{trouve2005metamorphoses,holm2009euler} is
available which introduces a parameter $\sigma^2$ parameterising the deviation
from diffeomorphic matching allowing for topological variations e.g. growth via
image intensity. In particular, if $\sigma^2=0$ the deformation is purely
diffeomorphic as in LDDMM.  See
\cite{trouve1995infinite,trouve2005local,miller2001group} for technical details
pertaining to the construction of the metamorphosis problem. While diffeomorphic
paths always exist for landmark problems \cite{guo2006diffeomorphic} this
theory allows one to match images of shapes with different topological features,
which is ill-conditioned for standard LDDMM. Indeed, even inexact matching in
LDDMM for such problems yields large energies and spurious geodesics that do not
contribute to an intuitive matching, see figure \ref{fig:mm_lddmm}. As observed
here, introducing $\sigma^2>0$ regularises the problem and qualitative improves
the matching.\\

In this work, we modify metamorphosis to include a spatially dependent control
parameter $x\mapsto\nu(x)$ in order to selectively allow non-diffeomorphic
(\emph{metamorphic}) matching in parts of the image. For $\nu(\cdot) = \sigma^2$
our theory recovers the standard metamorphosis model. However, with a localised
control (e.g. a Gaussian centred at a point in $\Rd$), we can selectively
introduce metamorphosis in an image and model local topological effects such as
growth phenomena. The difficulty of this problem is to infer the function
$\nu(\cdot)$ without prior knowledge of the location of the topological effects.
This problem is similar to the one treated in \cite{arnaudon_geometric_2017},
where such functions were parameterising the randomness in LDDMM matching of
shapes.  We will use a Markov chain Monte Carlo (MCMC) approach to infer
appropriate functions $\nu(\cdot)$, such that the topological effects are well
described and a large part of the matching remains diffeomorphic. In this paper,
we focus on landmark matching but aim to extend the theory of selective
metamorphosis to data structures amenable to classical metamorphosis or LDDMM
theory can handle.

\begin{figure}
\centering
\begin{minipage}{\textwidth}
  \centering
  \subfigure{\includegraphics[width=.3\textwidth]{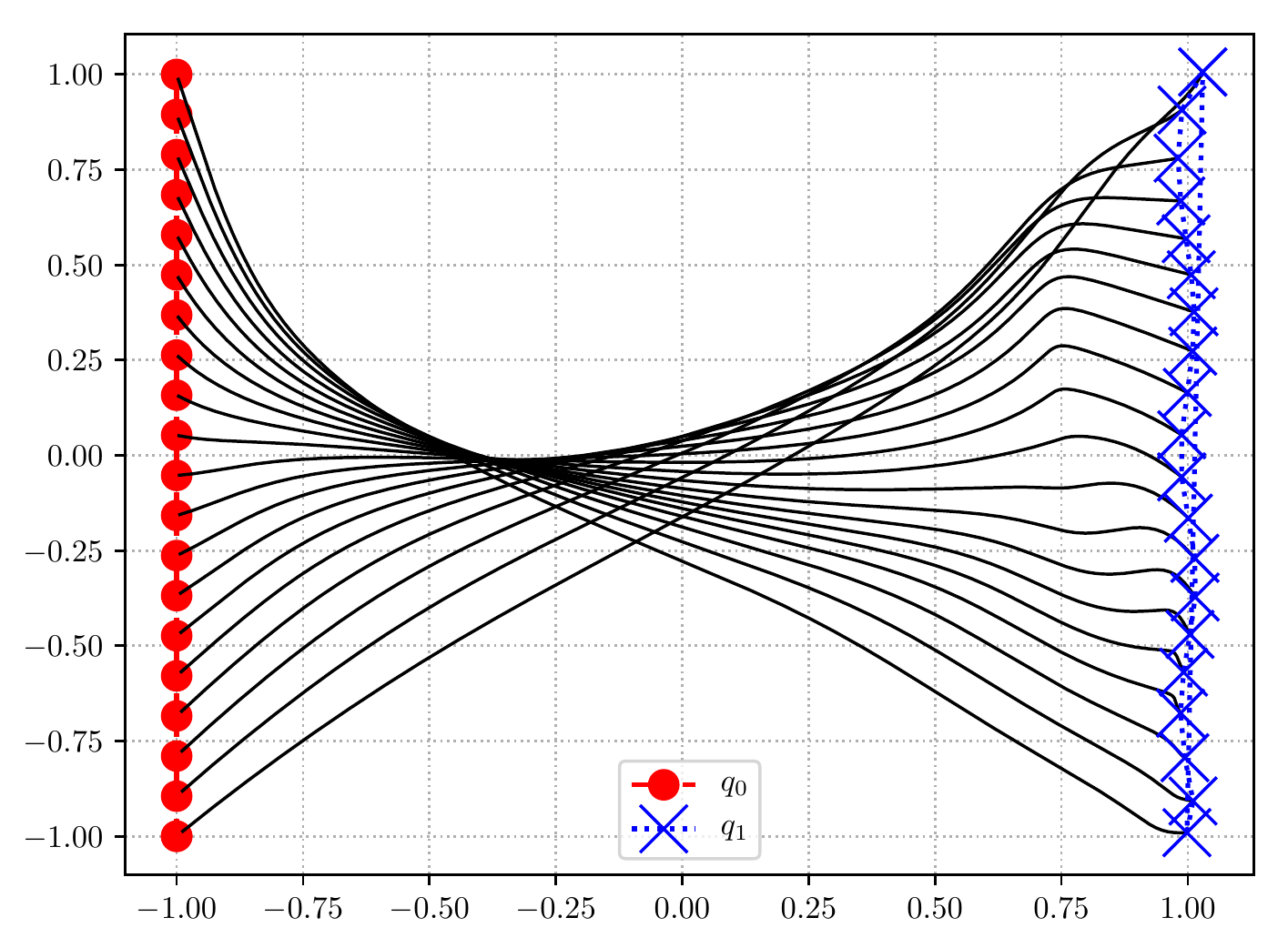}}
  \subfigure{\includegraphics[width=.3\textwidth]{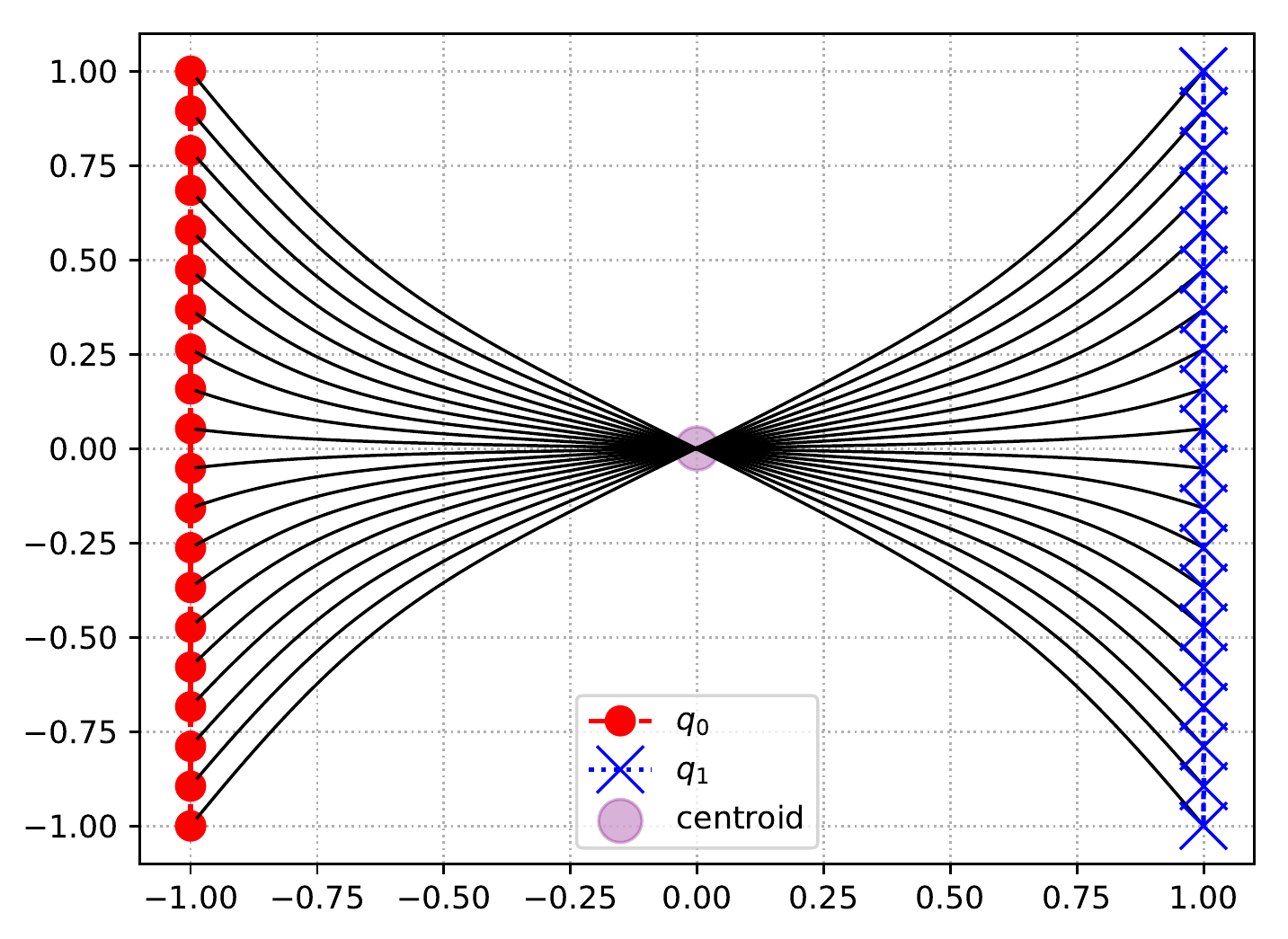}}
  \subfigure{\includegraphics[width=.3\textwidth]{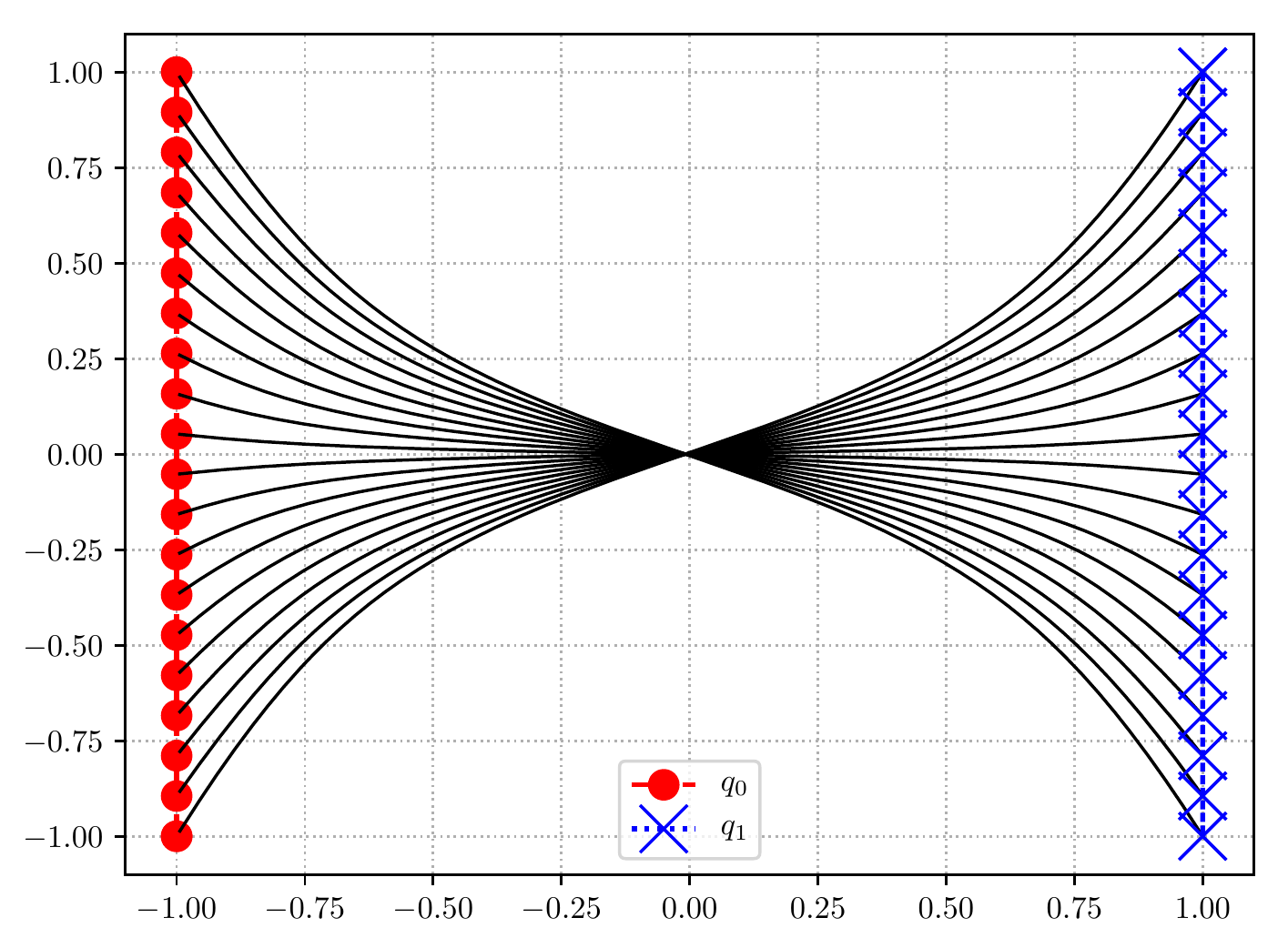}}
  \subfigure{\includegraphics[width=.3\textwidth]{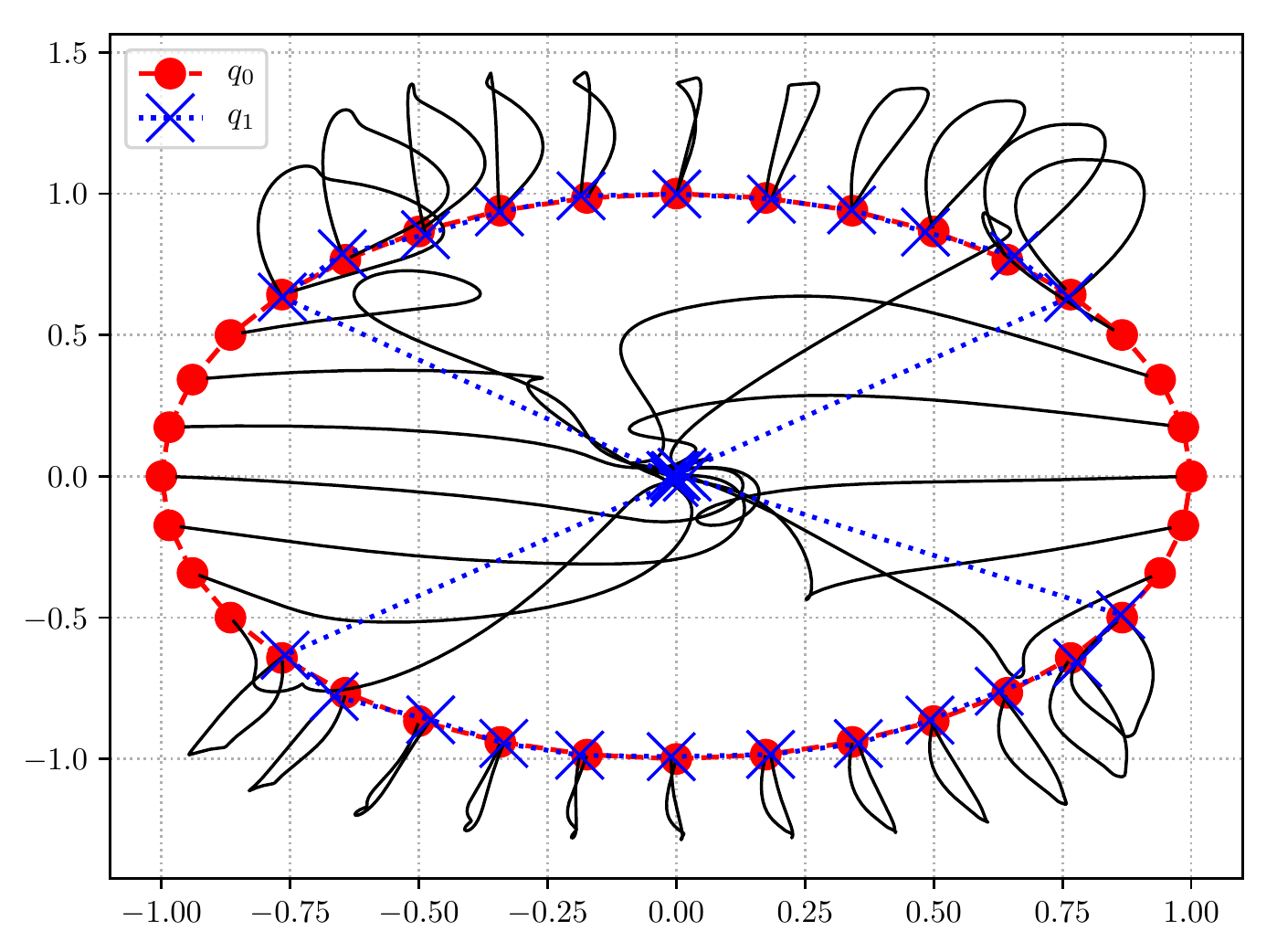}}
  \subfigure{\includegraphics[width=.3\textwidth]{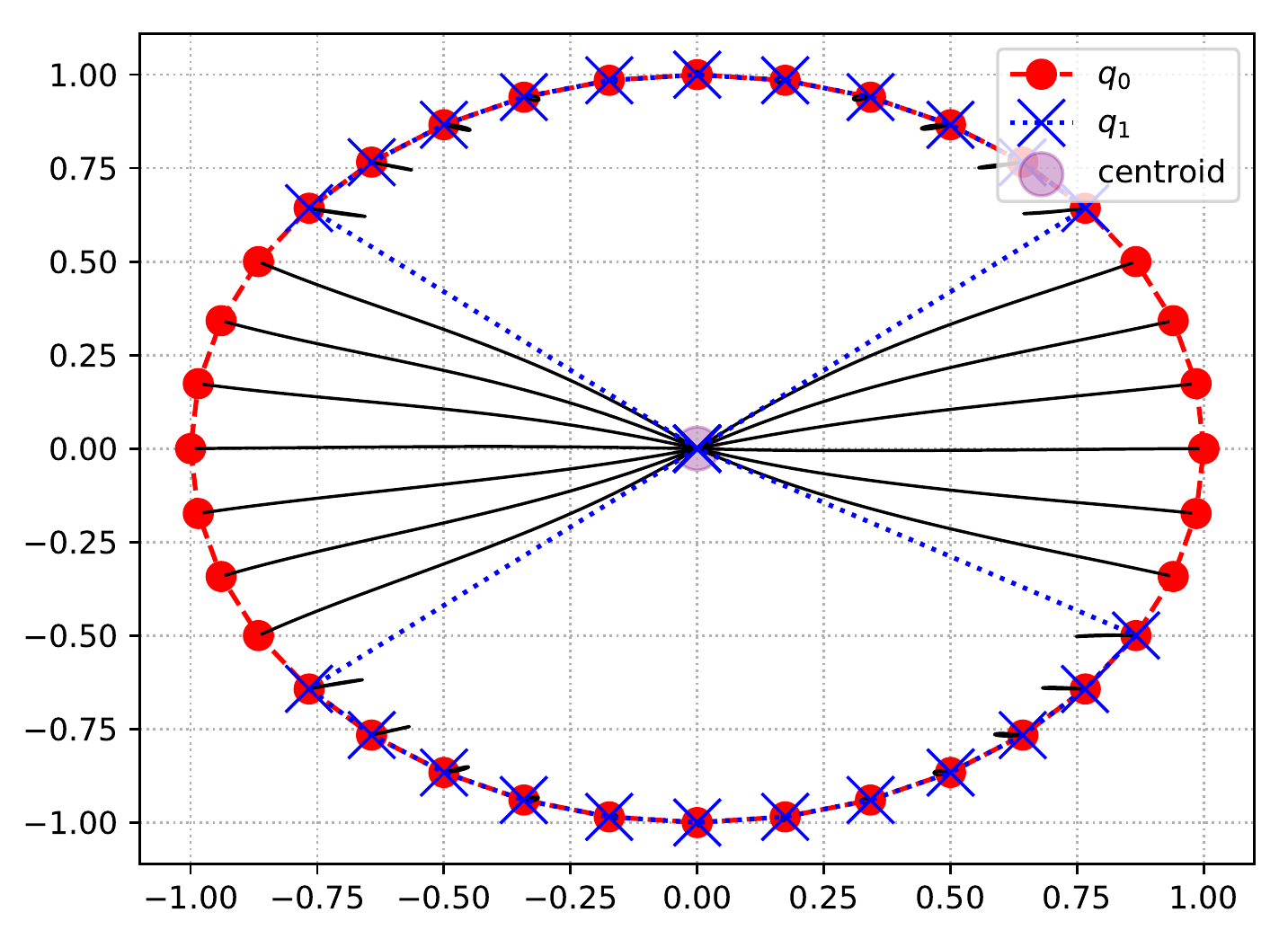}}
  \subfigure{\includegraphics[width=.3\textwidth]{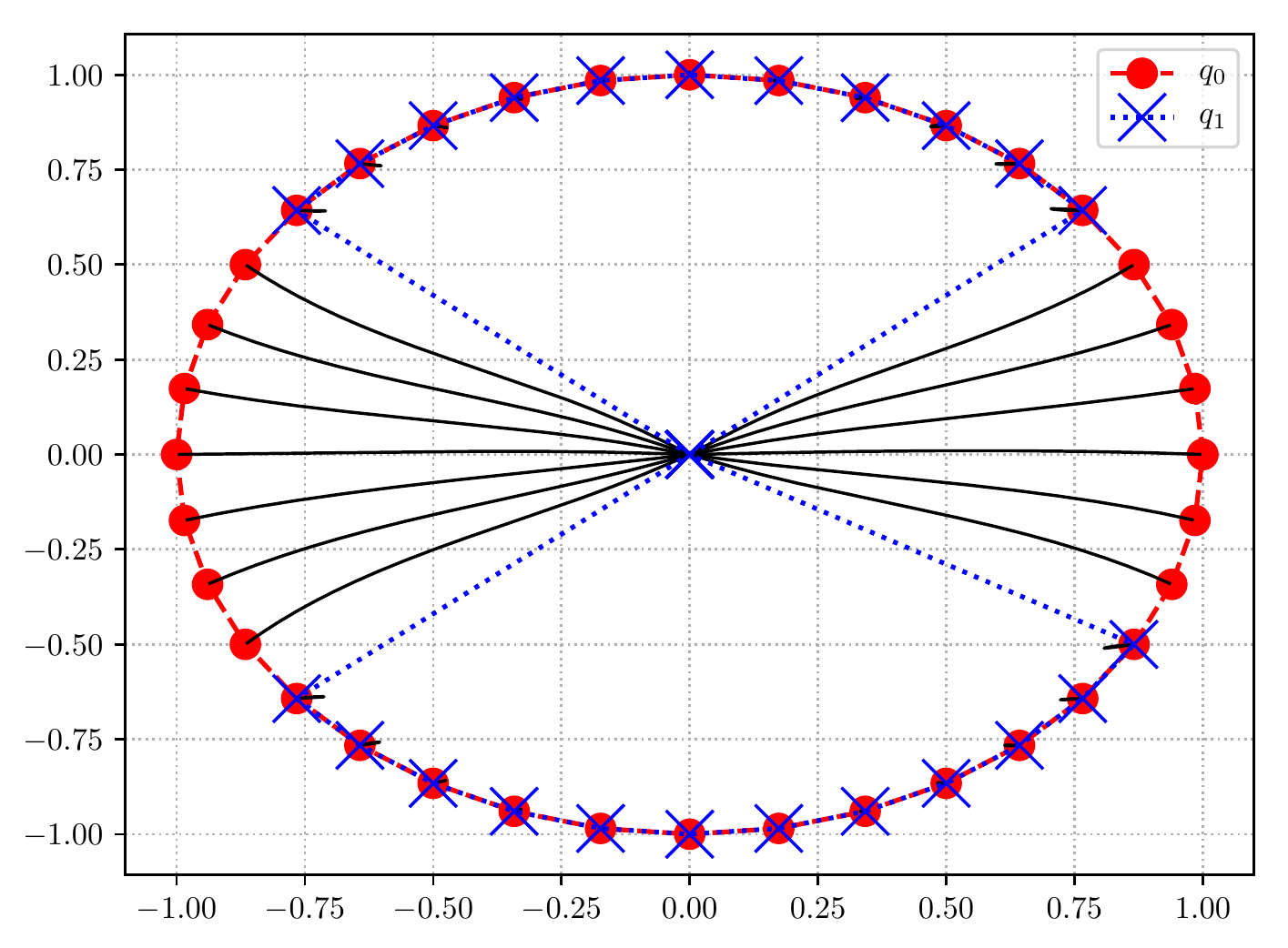}}
    \caption{This figure illustrates landmark matching with classical LDDMM
    (left column), metamorphosis (right column) and our selective metamorphosis
    approach (middle column). LDDMM fails to perform the matching and we observe
    unnatural landmark trajectories whereas metamorphosis achieves a more
    intuitive matching. Selective metamorphosis has the additional advantage of
    only breaking the diffeomorphic property where needed in along the matching,
    thus preserving more of the desired diffeomorphic property of the matching.
    These simulations where done for landmarks with Gaussian kernel of variance 
  $0.5$, $100$ timesteps from $t=0$ to $t=1$, and a metamorphosis kernel 
of variance $0.2$.}
    \label{fig:mm_lddmm}
\end{minipage}
\end{figure}

{\bf Structure} This paper is organised as follows. We review the theory of
classical metamorphosis in section \ref{sec:bg} and extend it to selective
metamorphosis in section \ref{sec:select_mm}.  We then introduce a Bayesian
framework for inferring the metamorphic control parameter $\nu(\cdot)$ in
section \ref{sec:bayesian} and apply this theory to a few landmark examples in
section \ref{sec:numerical}. Section \label{sec:outlook} contains concluding
remarks.

\section{Metamorphosis for Landmarks}\label{sec:bg}

In this paper we are concerned with diffeomorphometric approaches to image and
shape matching. To this end, we use time-dependent velocity fields $u_t$ occupying
some Hilbert space $u_t \in V$, where $V$ is continuously embedded in
$\textsf{C}_0^k(\Rd)$, $k\geq 1$ inducing a curve $\varphi_t$ on a subgroup
$\text{Diff}_V(\Rd)$ of diffeomorphisms
\cite{arnold1966geometrie,younes2010shapes} via the following ordinary
differential equation
\begin{align}
& \dot{\varphi_t} = u_t \circ \varphi_t\, , \qquad  \varphi_0 = \text{id}\, . 
  \label{diffeo}
\end{align}
This is often used in a minimisation problem where the objective is to match two
images $I_0$ and $I_1$:
\begin{align}
  S(u) = \int_0^1 \half\vnorm{u_t}^2 \diff{t} + \frac{1}{2\lambda^2}
  F(I_0\circ\varphi^{-1}_1, I_1)\longrightarrow \text{min. subject to \eqref{diffeo}}\, , \label{E-def}
\end{align}
where $F$ denotes a similarity measure between the deformed initial image
$I_0\circ \varphi_1$ and the target image $I_1$ to allow inexact matching
parameterised by $\lambda^2$. The LDDMM approach takes $F$ as an $L^2$ norm of
the difference between its arguments. In order to simplify the exposition, we
will consider singular solutions of this problem, which are given by the Ansatz
\begin{align}
  \frac{\delta l}{\delta u} = \mathbf m(x) = \sum_{i=1}^M p_t^i \delta(x-q_t^i)\,, 
\end{align}
for $M$ landmarks with position $q_t^i\in \Rd$ and momenta $p_t^i \in \Rd$.
The vector field is thus 
\begin{align}
  u_t(x) = \sum_{i=1}^M p_t^i K(x-q_t^i)\,, 
  \label{u-def}
\end{align}
where $K:\mathbb R^d\times \mathbb R^d\to \mathbb R$ is the kernel associated to
the norm $\|\cdot \|_V$. This parameterisation holds throughout this paper and
we set $d=2$. For metamorphosis, we introduce a discrete template variable
$\boldsymbol \eta_t$ such that the deformation of a set of landmarks is written
as the composition of the template position and deformation as
\begin{align}
  \mathbf q_t = \varphi_t \boldsymbol \eta_t\, . 
  \label{q_t}
\end{align}
Then, we can define the template velocity as 
\begin{align}
  \mathbf z = \varphi_t \dot {\boldsymbol \eta}
  \label{z_eta}
\end{align}
and extend the action functional \eqref{E-def} to 
\begin{align}
  \begin{split}
    S_m(\mathbf q_t, \mathbf p_t, \mathbf z_t) = & \int_0^1
    \half  \left (\vnorm{u_t}^2 + \frac{1}{\sigma^2} \sum_{i=1}^M |z_t^i|^2\right )\diff{t}\, , 
  \end{split}
  \label{E_m-def}
\end{align}
where now the reconstruction relation is 
\begin{align}
    \dot{\mathbf q_t} = u_t (\mathbf q_t) + \mathbf z_t\, , 
    \label{dq-m}
\end{align}
obtained from \eqref{z_eta} and \eqref{q_t} together with  $u= \dot \varphi_t \varphi^{-1}$, 
see \cite{holm2009euler} for more details.  

By taking variations carefully, see again \cite{holm2009euler}, we directly obtain a relationship
between the momentum variable $\mathbf p$ and the template variable $\mathbf z$ as
\begin{align}
  \mathbf m(x) = \frac{1}{\sigma^2} \sum_{i=1}^M z_t^i\delta(x-q_i)\qquad \Rightarrow \qquad
  \mathbf z = \sigma^2 \mathbf p_i\, , 
\end{align}
and the equation of motions are
\begin{align}
  \begin{split}
  \dot{\mathbf p_t} &= - \nabla u_t(\mathbf q_t)^T \mathbf p_t\\ 
  \dot{\mathbf q_t} &= u_t(\mathbf q_t) +  \sigma^2\mathbf p_t \,,
  \end{split}
  \label{eq-m-classic}
\end{align}
where $u_t$ is defined in \eqref{u-def}. 


\section{Selective Metamorphosis for Landmarks}\label{sec:select_mm}

We can now extend the metamorphosis setting to be able to locally control the
amount of non-diffeomorphic evolution.  For this, we introduce a function $\nu:
\mathbb R^2\to \mathbb R$ replacing the parameter $\sigma^2$ such that
$\nu(x)=\sigma^2$ corresponds to the classic landmark metamorphosis. The action 
for selective metamorphosis thus becomes 
\begin{align}
  \begin{split}
    S_{sm}^\nu(\mathbf q_t, u_t, \mathbf z_t) = & \int_0^1
    \half  \left (\vnorm{u_t}^2 +\sum_{i=1}^M \frac{1}{\nu(q_t^i)}|z_t^i|^2\right )\diff{t}\, , 
  \end{split}
  \label{E_sm-def}
\end{align}
which we minimise subject to the reconstruction equation \eqref{dq-m} and the
boundary conditions $\mathbf q_0$ and $\mathbf q_1$ at time $t=0,\,1$. In the
case of landmarks we have as before that
\begin{align}\label{zp_relation}
  \mathbf m(x) = \sum_{i=1}^M \frac{1}{\nu(q_t^i)} z_t^i\delta(x-q_i)\qquad \Rightarrow \qquad
  z_t^i = \nu(q_t^i) p_t^i\quad \forall i\, , 
\end{align}
so we can eliminate the template variable $\mathbf z_t$ and write  
\begin{align}
  \begin{split}
    S_{sm}^\nu(\mathbf q_t, u_t, \mathbf p_t) = & \int_0^1
    \half  \left (\|u_t\|_V^2  +\sum_{i=1}^M \nu(q_t^i)|p_t^i|^2\right )\diff{t}\, . 
  \end{split}
  \label{E_sm-def_p}
\end{align}

The problem defined by \eqref{E_sm-def_p} yields the following equations for
selective metamorphosis for landmarks:
\begin{align}
  \begin{split}
  \dot{\mathbf p}_t &= - \nabla u_t(\mathbf q_t)^T \mathbf p_t - \frac12
  \nabla \nu(\mathbf q_t ) |\mathbf p_t|^2\\ \dot{\mathbf q}_t &= u_t(\mathbf q_t) +
  \nu(\mathbf q_t)\mathbf p_t \,,
\end{split}
  \label{eq-m-selective}
\end{align}
with $\mathbf q_0,\, \mathbf q_1 \text{ fixed}$. Again, the velocity is fully
described by $\mathbf p$ and $\mathbf q$ via \eqref{u-def}. The landmark
dynamics follow standard LDDMM trajectories as $\nu(x)$ vanishes in parts of the
domain.  This can also be seen in the relation \eqref{zp_relation}, where
$\nu(x)=0 \Rightarrow \mathbf z_t=0$ implies that the template variable remains
fixed. Notice that these equations are Hamilton's equations for the Hamiltonian
$\mathbf p_t$ and $\mathbf q_t$
\begin{align}
  h(\mathbf q, \mathbf p) = h_l(\mathbf q_t,\mathbf p_t)+\frac12\sum_i\nu(q_i) |p_i|^2\,,  
\end{align}
where we have used \eqref{u-def} for the standard landmark Hamiltonian
\begin{align}
  h_l(\mathbf q, \mathbf p) =  \frac12 \sum_{i,j=1}^M K(q_i-q_j)p_i\cdot p_j\, . 
\label{hamiltonian}
\end{align}
A practical procedure for solving \eqref{eq-m-selective} with the landmark
Hamiltonian above is called \emph{shooting}, where we replace the end-point
condition $\mathbf q_1$ with a guess for $\mathbf p_0$, and iteratively update
$\mathbf p_0$ using automatically computed adjoint (or backward) equations
until $\mathbf q_1$ compares to $\mathbf q(1)$ below a certain tolerance. We
will perform this procedure directly with an automatic differentiation package
Theano \cite{team2016theano}, see
\cite{kuhnel2017computational,kuhnel2017differential} for more details on the
implementation. This section concludes with some theoretical results.
\begin{theorem}\label{sm-eu}
Let $\nu$ be bounded from below away from zero by $\nuinf \in \mathbb R$ and
from above by $0<\sigma^2\in\mathbb R$. Then there exists a minimiser of
\eqref{E_sm-def_p} admissible to \eqref{dq-m}.
\end{theorem}
\begin{proof}
The functional in \eqref{E_sm-def_p} is not convex so we work with a
reformulation to ensure the required lower semi-continuity. Define a
variable $w^i_t = \sqrt{\nu(q_t^i)} p^i_t$ in the problem:

\begin{align*}
\inf_{\substack{u \in L^2([0,1],\,V)\\ \mathbf q\in H^1([0,1],\,\RdM)\\\mathbf w\in L^2([0,1],\,\RdM)}}
    & \int_0^1\half\left (\|u_t\|^2_V + \sum_{i=1}^M |w_t^i|^2\right )\diff{t}\\
    & \dot{\mathbf q_t^i} = u_t(\mathbf q_t) + \sqrt{\nu(\mathbf q_t)} \mathbf w_t\\
    & \mathbf q_0,\,\mathbf q_1\text{ fixed}
  \label{pbl:reformulation}
\end{align*}
First, note that owing to
the constraint effectively being a boundary value problem, we cannot always find
a $\mathbf q$ for arbitrary pairs of $(u,\,\mathbf w)$. We define a bounded operator
$(\mathbf q,\, u_t)\mapsto \frac{\dot{\mathbf q_t} - u_t(\mathbf
q_t)}{\sqrt{\nu(\mathbf q_t)}} \triangleq \mathbf w$:
\begin{align*}
\Big(\sum_{i=1}^M |w_t^i|^2\Big)^{\frac 12} & = \ltwonorm{\mathbf w} =
\ltwonorm{\frac{\dot{\mathbf q}_t - u_t(\mathbf q_t)}{\sqrt{\nu(\mathbf q_t)}}}
\lesssim \nuinf^{-1}\Big(\ltwonorm{\dot{\mathbf q}_t} + \vnorm{u_t(\mathbf
q_t)}\Big)\,.
\end{align*}
From this we generate a minimising sequence $(\mathbf q^n, u^n, \mathbf
w^n)_{n\geq 0}$ admissible to \eqref{pbl:reformulation}. The rest of the proof
is standard, see e.g. \cite{younes2010shapes}. We show the constraint equation
is continuous with respect to the weak topology on $X\triangleq
H^1([0,1],\,\RdM)\times L^2([0,1],\,V)\times L^2([0,1],\,\RdM)$ i.e.  $e(\mathbf
q_t^n,\,\mathbf w_t^n,\, u_t^n)\rightharpoonup e(\mathbf q_t,\,\mathbf
w_t,\, u_t)$ where $e(q,\,w,\,u) \triangleq \dot{q} - u(q) -
\sqrt{\nu(q)}w$. Then,
\begin{align*}
\langle\sqrt{\nu(\mathbf q_t)}\mathbf w_t - \sqrt{\nu(\mathbf q_t^n)}\mathbf w_t^n,\,\phi\rangle \lesssim \nuinf\langle
\mathbf w_t - \mathbf w_t^n,\,\phi\rangle \rightarrow 0\,,\quad \forall \phi \in
L^2([0,1],\,\RdM)\,.
\end{align*}
Further, for $\phi\in L^2([0,1],\,V)$,
\begin{align*}
\langle u_t(\mathbf q_t)-u_t^n(\mathbf q_t^n),\,\phi\rangle & = \langle
u_t(\mathbf q_t)- u_t^n(\mathbf q_t),\,\phi\rangle + \langle u_t^n(\mathbf
q_t)-u^n(\mathbf q_t^n),\,\phi\rangle\,.
\end{align*}
The first term vanishes trivially, while for the second we see
\begin{align*}
\langle u_t^n(\mathbf q_t)-u_t^n(\mathbf q_t^n),\,\phi\rangle \leq \text{Lip}(u_t^n)\langle \mathbf q_t-\mathbf q_t^n,\,\phi\rangle \rightarrow 0
\end{align*}
Since linear operators are naturally compatible with the weak topology the
required continuity follows.  Passing to subsequences where necessary we can by
classic results extract bounded subsequences converging to weak limits where
necessary to obtain a minimiser. Convexity of $S$ implies weak lower
semi-continuity concluding the proof.
{\hfill $\square$}
\end{proof}
\begin{theorem}
Assume $\nu \in W^{2, \infty}(\Rd)$ and $V$ is embedded in
$\textsf{C}_0^k(\Rd)$, $k\geq 1$ (continuous functions with continuous
derivatives to order $k$ vanishing at infinity). Then, given $\mathbf
p_0,\,\mathbf q_0, \in \RdM$, \eqref{eq-m-selective} with \eqref{u-def} are
integrable for all time.
\end{theorem}
\begin{proof}
Establishing appropriate Lipschitz conditions implies integrability of the
system akin to \cite[Theorem 5]{cotter2013bayesian}. We note that the kernel in
\eqref{u-def} is Lipschitz in $(p_t, q_t)$ by assumption, so the composition $(p,\,q) \mapsto
u\circ q$ is also Lipschitz. $u(q)\mapsto\nabla u(q)^T$ consider $v,\,w \in V$ and $x,\,y\in\Rd$:
\begin{align}
\ltwonorm{\nabla v(x) - \nabla w(y)} \lesssim \vnorm{v} \ltwonorm{x-y} + \vnorm{v-w}\ltwonorm{y}
\end{align}
so the mapping is Lipschitz in both the position and velocity. Given the
conditions on $\nu$ the mappings
\begin{align}
  \begin{split}
    & (q,\,p)\mapsto \nu(q)p\\
    &(q,\,p)\mapsto \nabla \nu(q)|p|^2
  \end{split}\label{nu_maps}
\end{align}
are locally Lipschitz. Consequently we verify that for any $(\mathbf
p_0,\,\mathbf q_0)\in B(0,r)\subset \RdM\times\RdM$, the system
\eqref{eq-m-selective} is locally Lipschitz with constant $L_{r,t_0}$ for some
$t_0>0$. By the conservation of the Hamiltonian we can extend the existence of
solutions to arbitary $t>t_0$.
{\hfill $\square$}
\end{proof}

\section{Bayesian Framework}\label{sec:bayesian}

We now place a stochastic model on $\nu$ inspired by the approach taken in
\cite{cotter2013bayesian}, see also
\cite{schiratti2017bayesian,allassonniere2007towards,allassonniere2008map} for
similar Bayesian approaches in computational anatomy. The goal is to develop an
algorithm to infer $\nu$, passing via the deterministic problem seen above.
First, we present some preliminaries on the Bayesian approach to inverse
problems in section \ref{subs:gf}, essentially quoting concepts and results from
\cite{dashti2017bayesian}, or \cite{cotter2013mcmc} for an exposition of
algorithmic aspects of function space MCMC. Section \ref{subs:finite-dim-param}
then formally describes how we apply this stochastic approach to inverse
problems to $\nu$ by a finite-dimensional family of parameterisations.

\subsection{General Framework}\label{subs:gf}

The general framework is based on the idea that we can cast optimisation
problems in a probabilistic framework where minimisers, roughly speaking,
correspond to modes of a certain distribution over function space.  In the
context of optimisation we define a \emph{likelihood} $\Phi : X\rightarrow Y$,
mapping some control variable in $X$ to an observable in $Y$.  A \emph{maximum a
posteriori} (MAP) estimator $f^*$ satisfies $f^* = \arg\max_{f\in X} \Phi(f)
p(f)$ where $p$ is a density over $X$. Equipped with a norm $\|\cdot\|_X$ we can
then define the Gaussian density by $p(\cdot) \propto e^{-\|\cdot\|_X^2}$.
Supposing further that the likelihood is on the form $\log \Phi(f) =
-\|f-\lambda\|_Y^2$, for some desiderata $\lambda\in Y$ then, at least formally,
the MAP estimator minimises $J=\|f^*\|_X^2 + \|f^*-\lambda\|_Y^2$.\\

For general inverse problems on function space, several key properties
documented in \cite{dashti2017bayesian} must be verified before the
inverse problem is well-posed. Beyond showing the existence of the MAP estimator
minimising $J$ above, the infinite-dimensional version of Bayes' rule must also
be checked i.e. the Radon-Nikodym derivative of the prior with respect to the
posterior must exist and be absolutely continuous. Finally, we request
continuity of the posterior distribution w.r.t the initial data corresponding in
a sense to Hadamard well-posedness in a probabilistic framework. Rigorously
treating this Bayesian inverse problem is subject to further study in 
forthcoming works.

\subsection{Finite-dimensional Parameterisation}\label{subs:finite-dim-param}

We now introduce the main problem of this paper in the setting above. We
consider $\nu$ as a random variable. If the growth location is not known \emph{a
priori}, then this is an appropriate framework as it allows for a qualitative
evaluation of selective metamorphosis. Moreover, it can account for some
observation error by its probabilistic nature. As a simple example, we consider
the case where $\nu$ is given by a sum of Gaussians 
\begin{equation}\label{nu_h}
    \nu_h (x) = \sum_{k=1}^K e^{ -\sigma_\nu^{-2}\|h_k - x\|_{\Rd}^2}\, . 
\end{equation}
Here the finite family $h_k$ of centroids in $\Rd$ together with the uniform
length-scale $\sigma_\nu$ fully determine $\nu_h$, thus greatly reducing the
complexity of sampling. We defer sampling from function space to future work.
Defining a density $p_{sm} \propto e^{- S_{sm}^\nu}$ over the space of triples
$(\nu,\,\mathbf q_\nu,\,\mathbf p_\nu)$ leads to the preconditioned
Crank-Nicholson algorithm \ref{algo:mcmc}, see e.g. \cite{hairer2014spectral}.
\begin{algorithm}[h!]
\begin{algorithmic}
\caption{MCMC for selective metamorphosis}\label{algo:mcmc}
\Procedure{mcmcSM}{$N$, $K$, $\mathbf q_0$, $\mathbf q_1$, $\beta\in (0,1]$}
\State $j \gets 1$
\State $\nu^j \gets \text{initial guess in } \RdK$
\State Solve \eqref{eq-m-selective} with $\nu^j$ and $\mathbf q_0,\,\mathbf q_1$
to obtain $\omega^j = (\mathbf q^j,\,\mathbf p^j,\, u^j)$
\While{$j < N$}
\State Sample a random point $\xi \in \mathcal N(0, \text{Id}_{\mathbb R^d})^K$
\State $\nu \gets \beta \xi + \sqrt{1-\beta^2} \nu^j$
\State Solve \eqref{eq-m-selective} with $\nu$ and $\mathbf q_0,\,\mathbf q_1$
to obtain $\omega = (\mathbf q,\, \mathbf p,\, u)$
\If {\textsc{randomUnit()}$\,< \min(1,\, e^{- S_{sm}^{\nu^j}(\omega^j) + S_{sm}^{\nu}(\omega)})$}
    \State $\nu^{j+1} \gets \nu$
    \State $\omega^{j+1} \gets \omega$
\Else
    \State $\nu^{j+1} \gets \nu^j$
    \State $\omega^{j+1} \gets \omega^j$
\EndIf
\State $j\gets j+1$
\EndWhile
\Return $\{\nu^j,\, \omega^j\}_{j=1}^N$
\EndProcedure
\end{algorithmic}
\end{algorithm}

Here, $N$ denotes the desired number
of samples and $K$ the number of terms in \eqref{nu_h}. \textsc{randomUnit()}
denotes a randomly generated number in $[0,\,1]$. The coefficient $\beta$ scales
between the previous sample and the new step $\xi$. This needs to be calibrated
as a too low value may increase the acceptance rate by taking shorter steps at
the cost of slow exploration of the state space. Conversely, a too high value of
$\beta$ results in lower acceptance and thus convergence. The next section shows
the algorithm above in practice.

\section{Numerical examples}\label{sec:numerical}

This section displays some numerical results for our method. We apply algorithm
\ref{algo:mcmc} to infer a distribution for the growth location using the
landmark boundary conditions seen in figure \ref{fig:mm_lddmm}.  The parameters
and results for the first configuration is shown in figure
\ref{fig:selective:crisscross}. These preliminary results demonstrate that even
for a small number of samples the density of accepted samples corresponds at
least heuristically to the analytical density histogram obtained by computing
the value of the metamorphosis functional in \eqref{E_sm-def}.

\begin{figure}[h!]
  \centering
  \subfigure{\includegraphics[scale=.35]{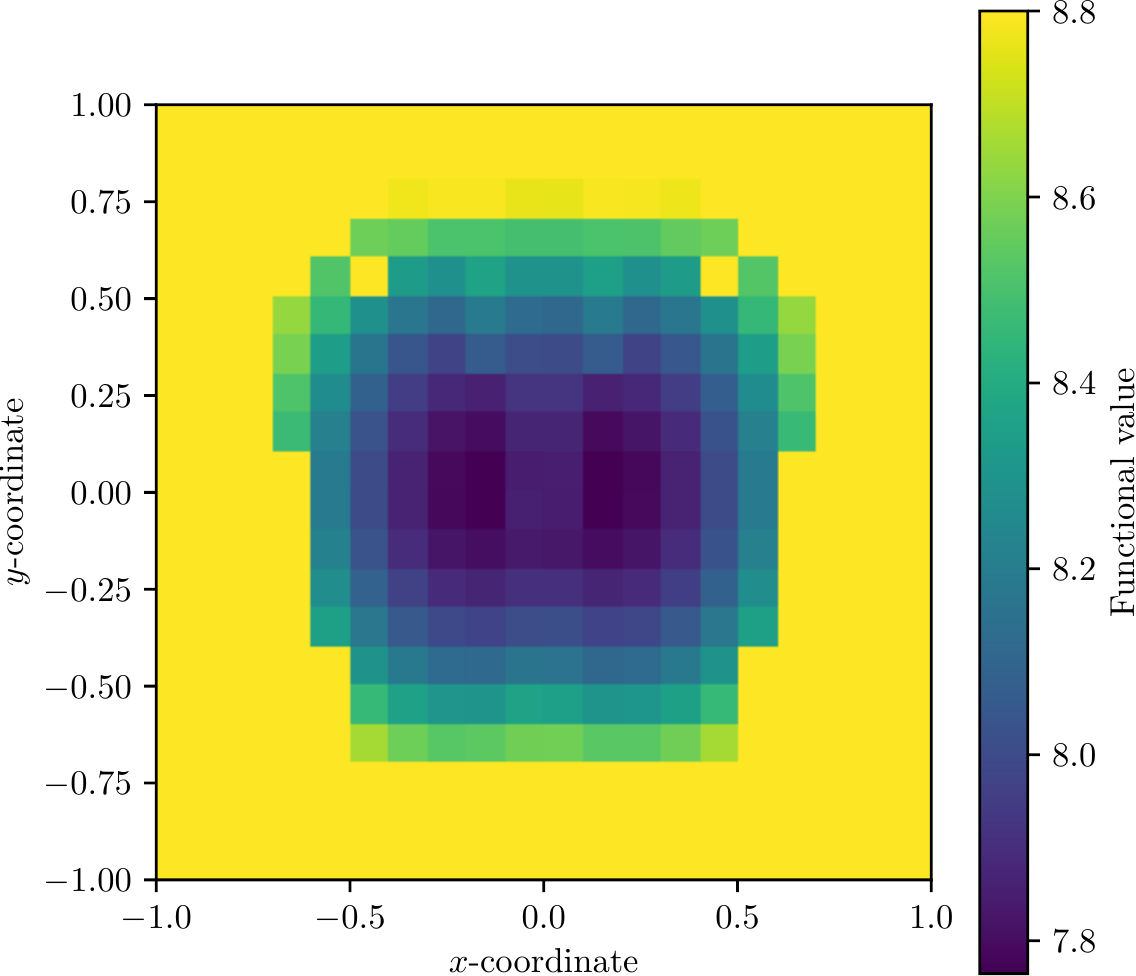}}
  \subfigure{\includegraphics[scale=.35]{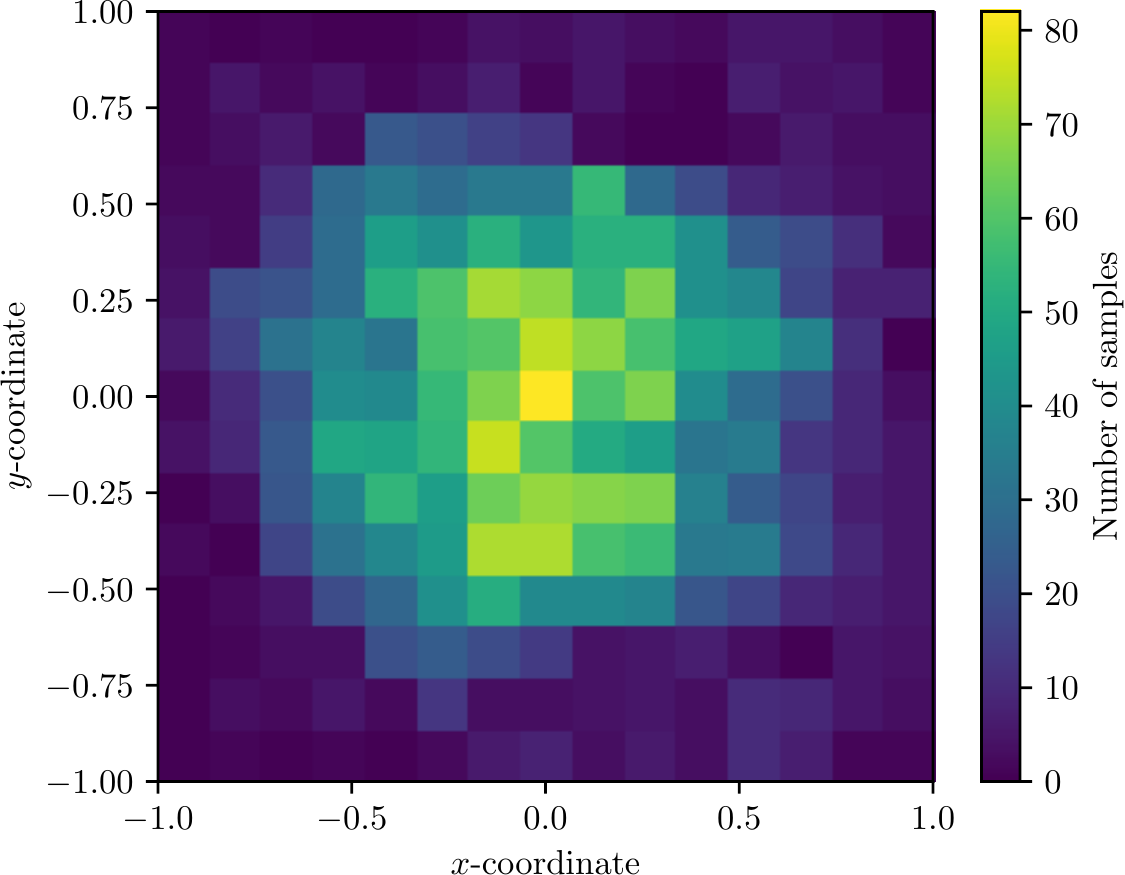}}
  \subfigure{\includegraphics[scale=.35]{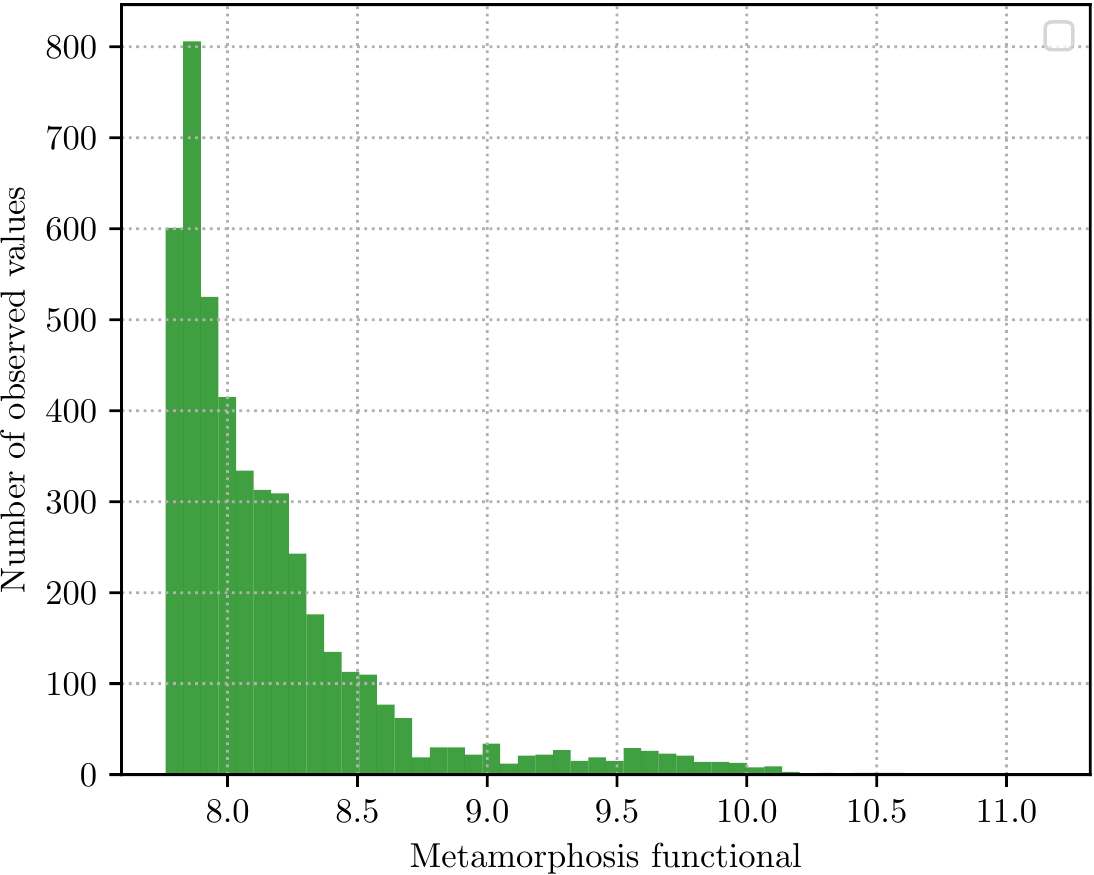}}
  \subfigure{\includegraphics[scale=.39]{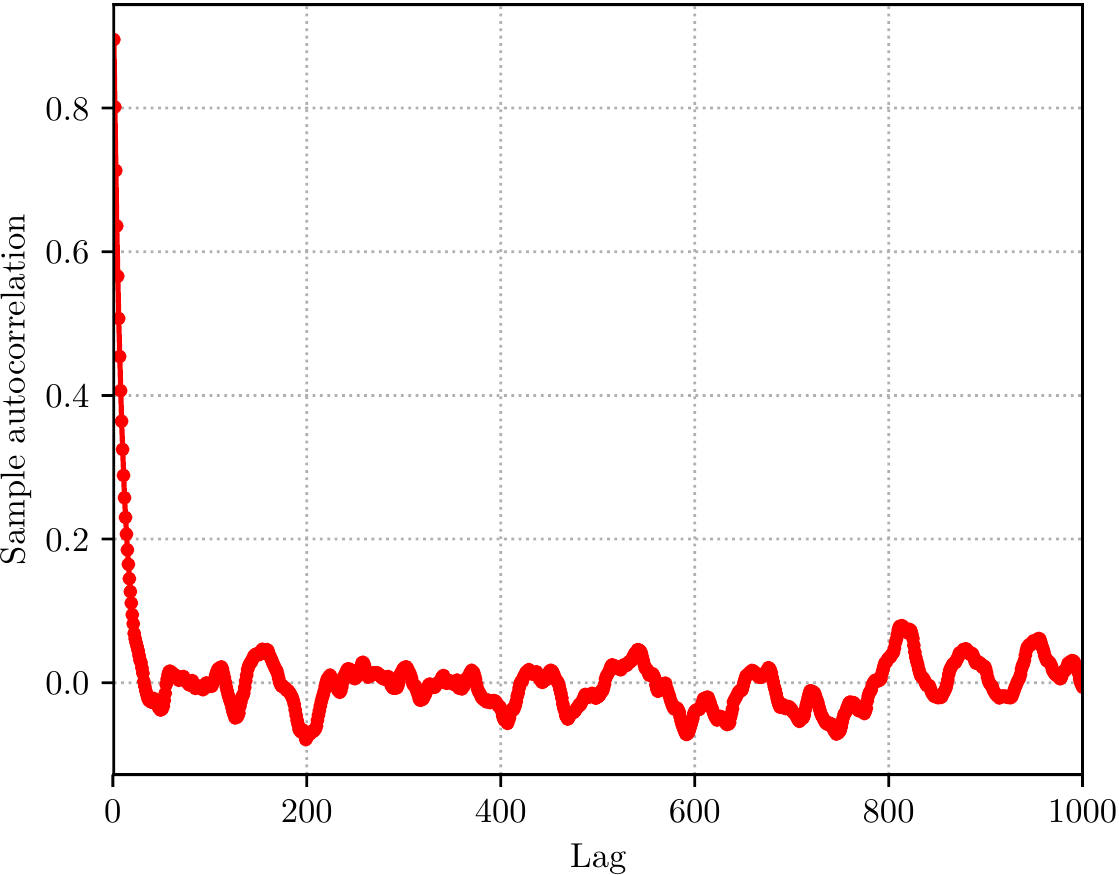}}
  \subfigure{\includegraphics[scale=.34]{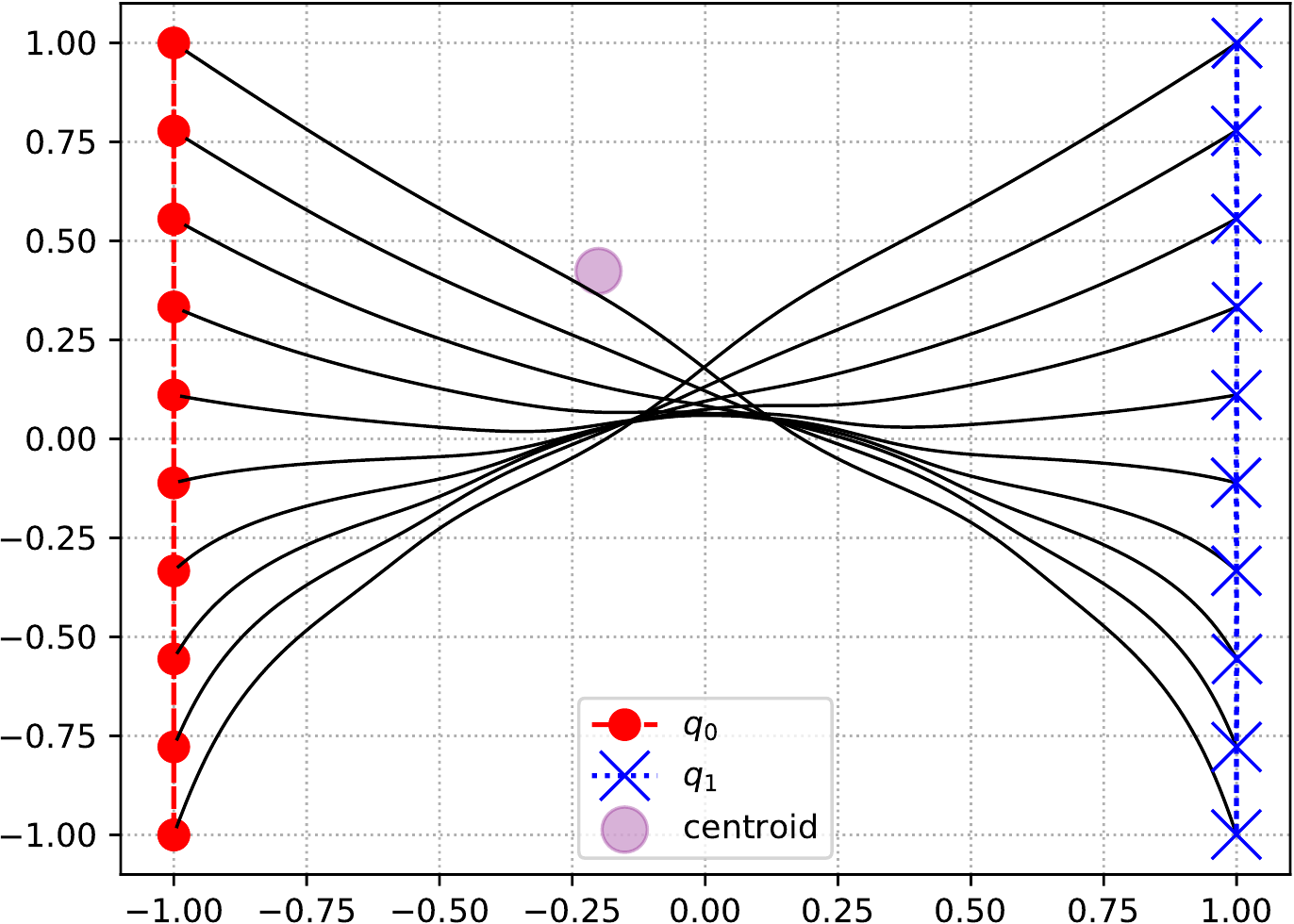}}
    \caption{We display the result of the MCMC algorithm \ref{algo:mcmc} applied
    to the inverted landmarks example of figure \ref{fig:mm_lddmm}.  The top left
    panel shows the analytical values for the functional \eqref{E_sm-def}
    obtained for various positions of a single Gaussian $\nu$. We observe a
    bimodal minimum near $(0,0)$, which depends on the choice of the model
    parameters, and in particular on the landmark interaction length
    corresponding to the Gaussian kernel $K$ and $\sigma_\nu$.  The top middle
    panel displays a heat map for the sampled positions of the centroid from the
    MCMC method, where the bimodality is not clearly visible.  The top right
    panel is a histogram of the sampled
    values of the functional which rapidly decays, indicating a good sampling of
    the minimum value of the functional.  The bottom left panel shows the
    autocorrelation function of the Markov chain, which decays rapidly to reach
    an un-correlated state after $50$ iterations.  The bottom right panel is one
    of the MAP estimators where the centroid is near on the edge of one of the
    wells of the top left panel.  The simulations parameters are set to
    $\sigma_\nu = 0.2$, and $0.7$ for the velocity kernel, $K=1$ and $\beta=0.2$
    across $5000$ samples.}
    \label{fig:selective:crisscross}
\end{figure}

We arrive at the same conclusion for the second example, for which the results
are shown in figure \eqref{fig:selective:pinch}. Moreover, we note that the
geodesic equations for $p$ and $q$ are time-reversible meaning that the
configuration in figure \ref{fig:selective:pinch} corresponds to both particle
collapse as well as hole creation. 

\begin{figure}[h!]
\centering
  \subfigure{\includegraphics[scale=.34]{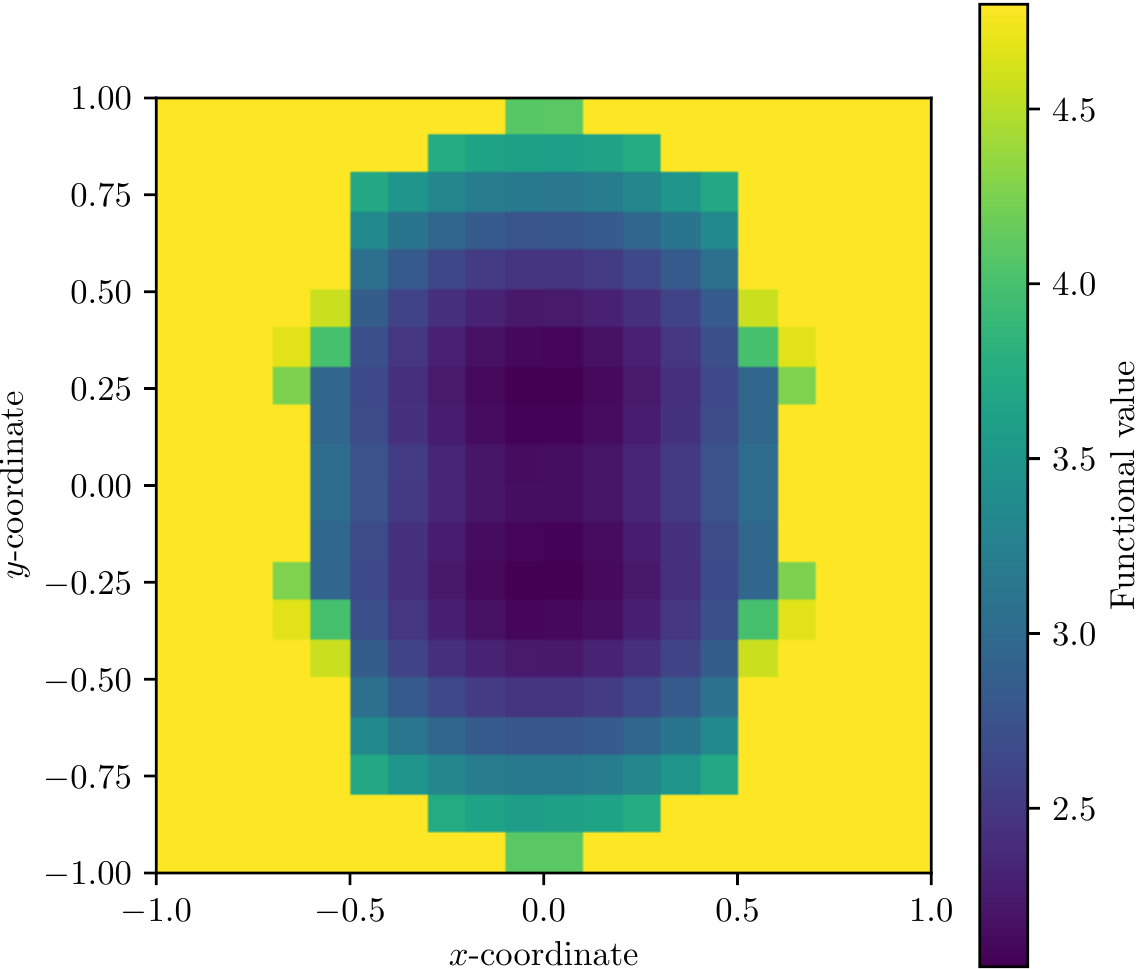}}
  \subfigure{\includegraphics[scale=.34]{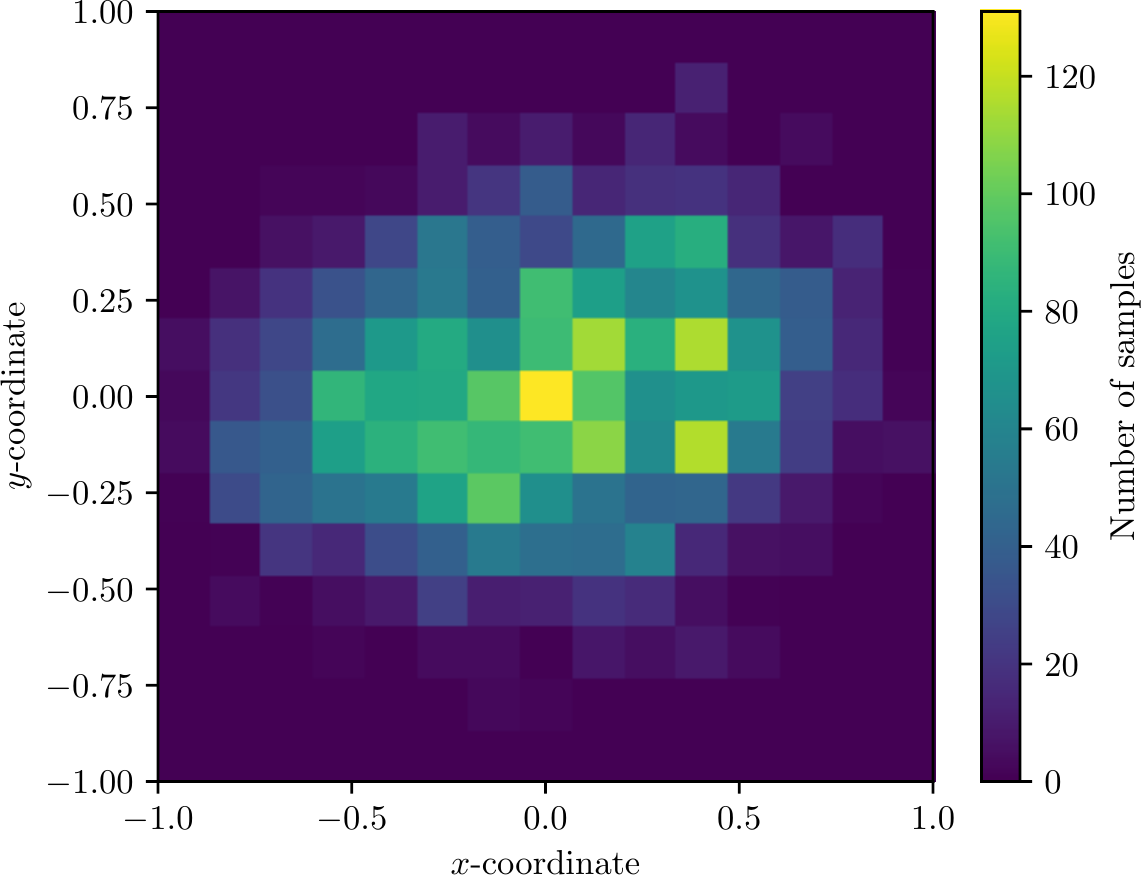}}
  \subfigure{\includegraphics[scale=.34]{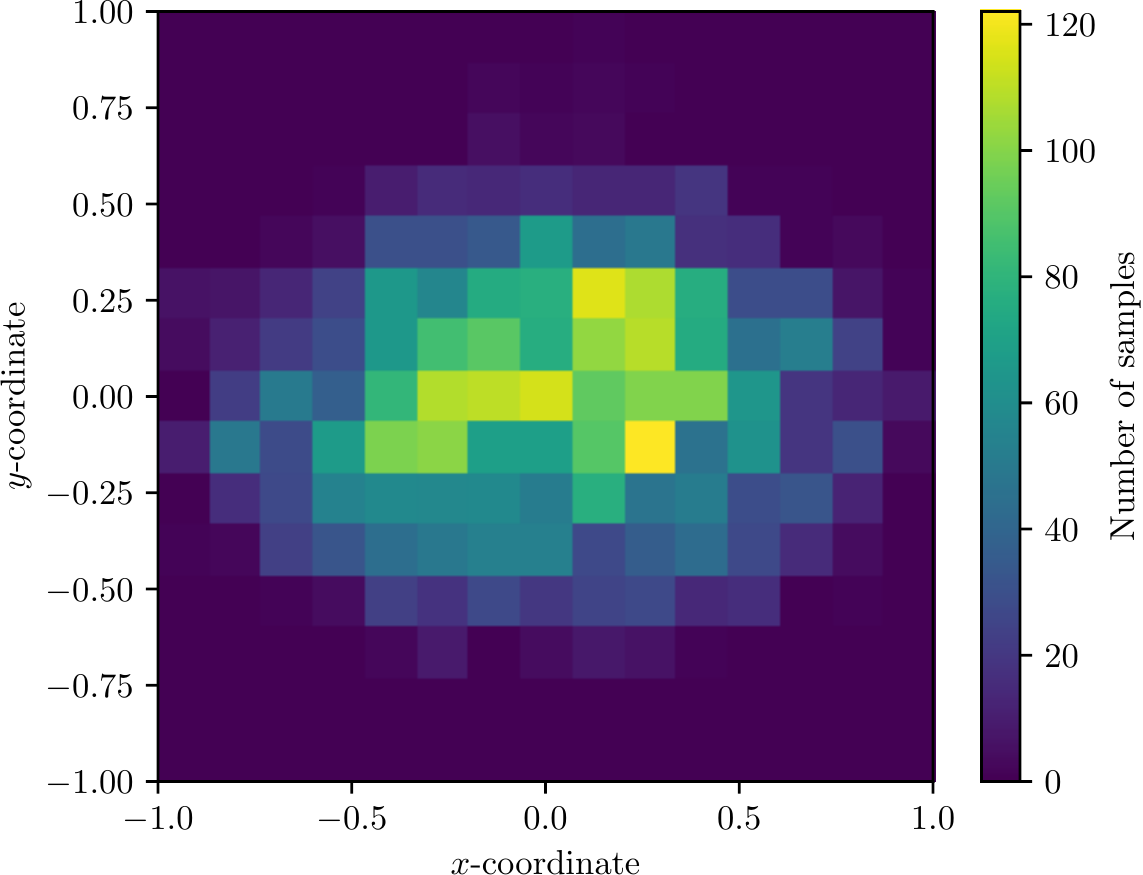}}
  \subfigure{\includegraphics[scale=.33]{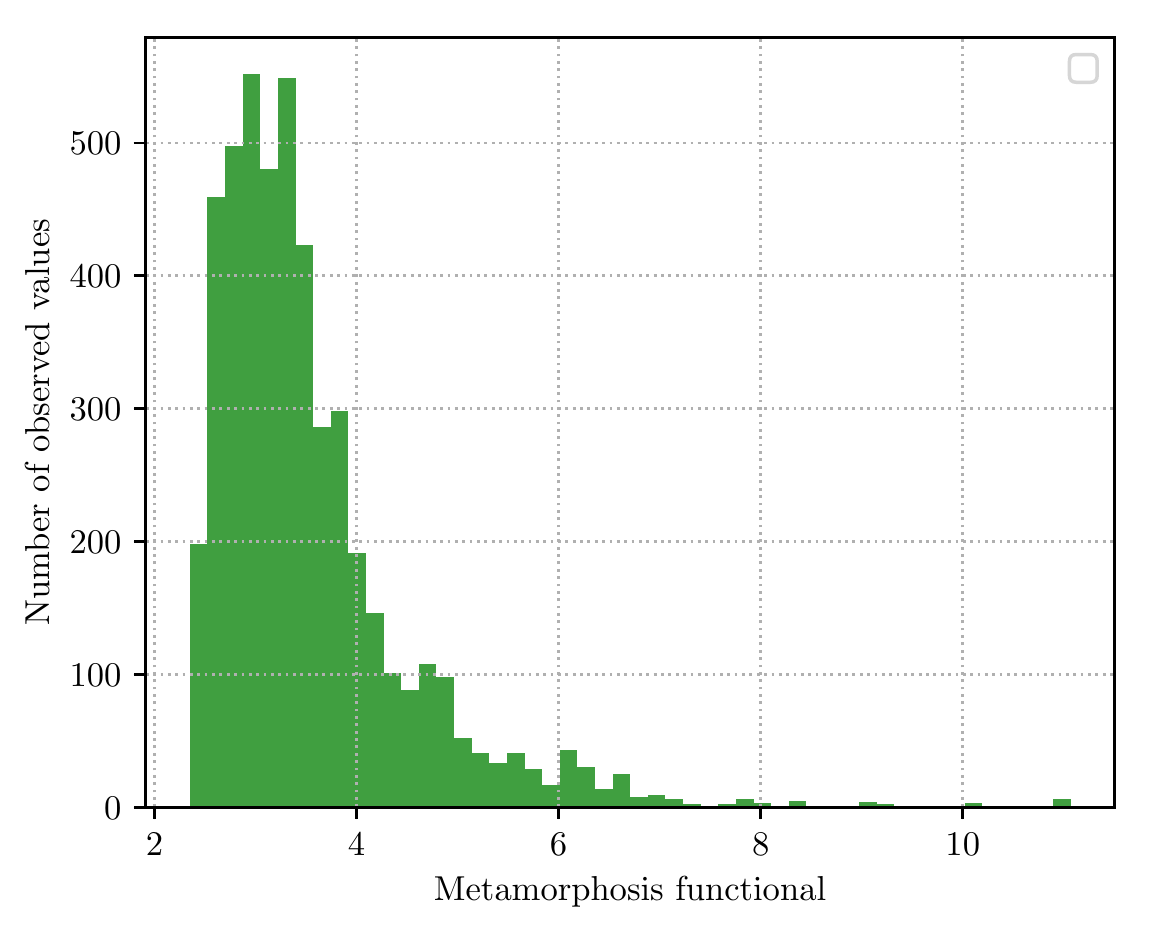}}
  \subfigure{\includegraphics[scale=.34]{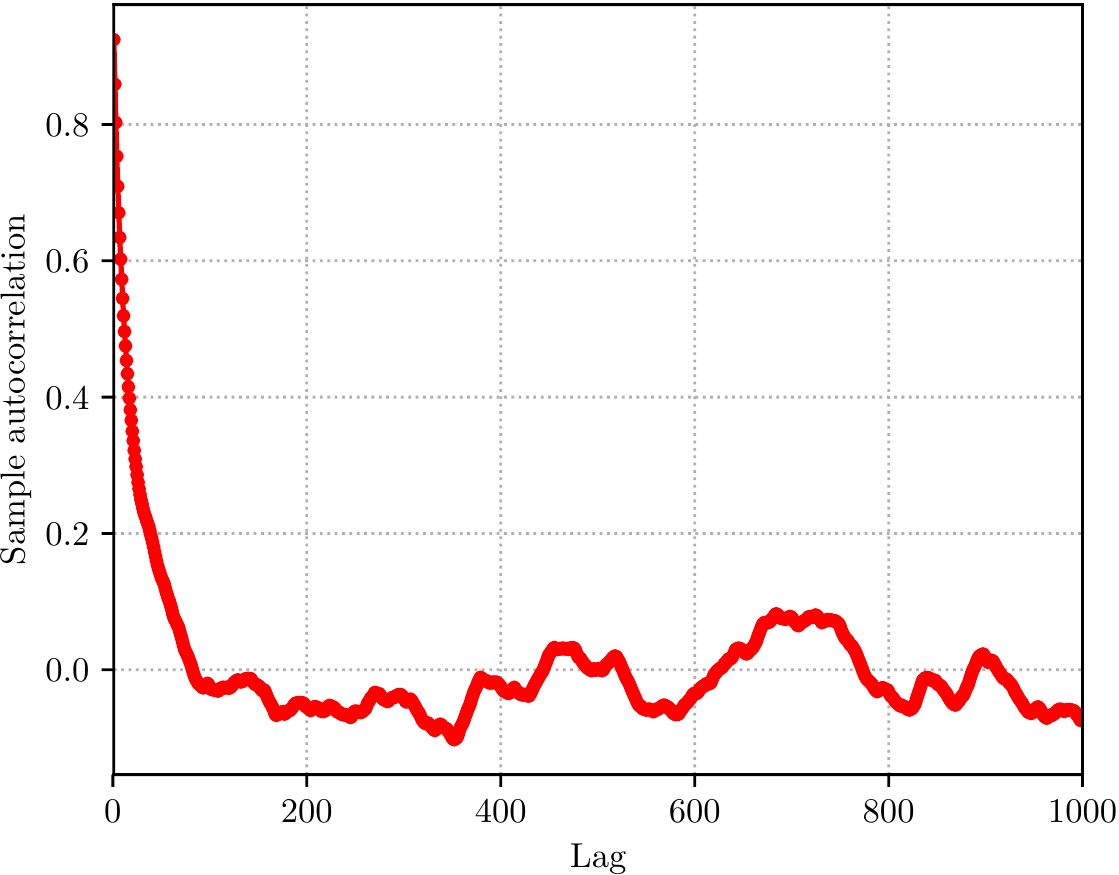}}
  \subfigure{\includegraphics[scale=.30]{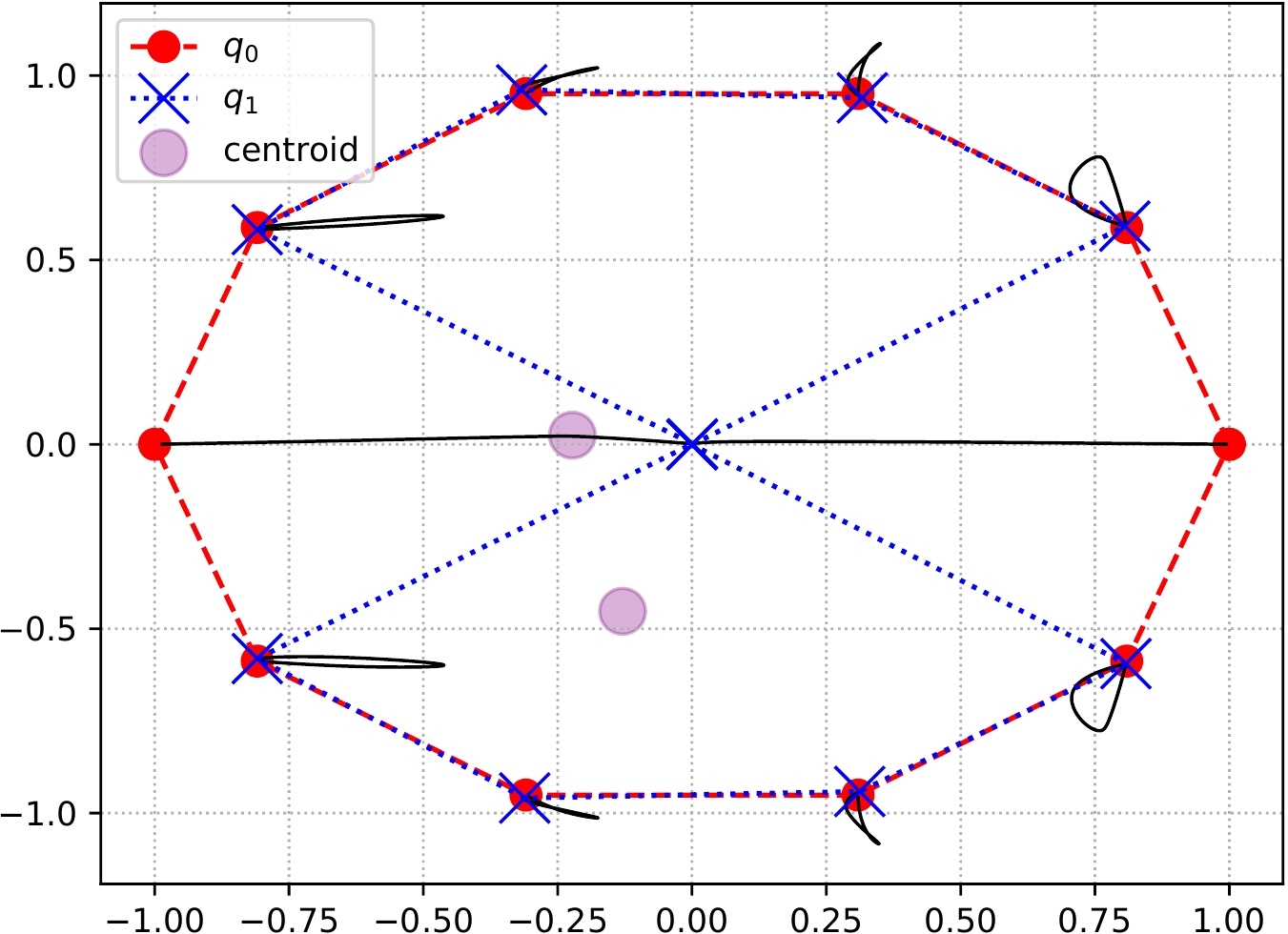}}
   
    \caption{ Here we display the results for the second example (landmark
    collapse) of figure \ref{fig:mm_lddmm}.  Again, the top left panel shows the
    analytical values for a single $\nu$ field \eqref{E_sm-def}, which has also
    a bimodal structure, but in the other direction.  For the MCMC we choose
    $K=2$ Gaussian $\nu$ fields, and the top middle and right displays two heat
    maps for the sampled positions of these centroids.  The bottom left panel is
    a histogram of the sampled values of the functional, which has a peak at
    slightly higher values, possibly due to the redundant choice of two $\nu$
    functions.  The bottom middle panel is the autocorrelation function of the
    Markov chain which shows decorrelation after $100$ steps.  The bottom right
    panel shows the geodesics yielding one of the lowest functional values,
    where the two $\nu$  fields are close to each other, demonstrating the fact
    that only $1$ would have been enough for this landmark configuration.  The
    simulation parameters are the same as in figure
    \ref{fig:selective:crisscross} with the exception of $K=2$.
}
    \label{fig:selective:pinch}
\end{figure}

It is numerically relatively simple to control the behaviour of $\nu$ by simple
scaling or by adding regularisation terms to \eqref{E_sm-def} to e.g. penalise
having $\nu$'s far away from the support of the images. 
Such cost can easily be added to the MCMC algorithm, depending on the prior
information one can have on the shape matching problem.

\section{Conclusion}\label{sec:outlook}

We have presented a preliminary approach for selectively allowing photometric
variation in a diffeomorphic image matching. We analysed the selective
metamorphosis problem, the associated geodesic equations and demonstrated a
proof of concept MCMC algorithm inferring a simple parameterisation of $\nu$.
This generalises LDDMM and metamorphosis and could provide a first-order
exploratory tool for physicians to see if the development of a biological
feature stems from a few violations of diffeomorphic evolution. This paper paves
the way towards surgically investigating growth phenomena between topologically
different images.

Future work is manifold. Firstly, we aim to extend the equations of section
\ref{sec:select_mm} to images e.g. using the kernel framework in
\cite{richardson2016metamorphosis} or developing a space-time method. We also
aim to find an explicit solution to the geodesic equations for $\mathbf p$ and
$\mathbf q$ and with the additional terms involving $\nu$ à la
\cite{trouve2005local} to eliminate the need $\nu$ to be bounded from below.
Further, as outlined in \ref{sec:bayesian} there
are many aspects of the probabilistic framework that need rigorous treatment.
Beyond the references therein, see also the work in \cite{dashti2013map}.
Natural extensions of our probabilistic approach also include fully treating
$\nu$ as a function and interpreting the resulting inverse problem through the
appropriate measure-theoretical lens. Adding a time-dependency to $\nu$ can also
be explored. Determining a truncated Fourier series of $\nu$ could lead to
efficient numerical methods. More generally, we hope to reconcile our attempts
to model growth here with the mathematically elegant approach described in
\cite{kaltenmark2016geometrical} and with the more general mathematics of growth
\cite{goriely2017mathematics}.

Finally, it is our hope that we can extend the probabilistic approach developed
here to encompass classic metamorphosis as well; that is to say, to develop the
necessary theory in order to place a stochastic model on the state space
consisting of velocities and source functions and sample from function space.
This provides a derivative-free method of solving classic metamorphosis (or
other problems in shape analysis) at the expense of interpreting the results
probabilistically.
In this work, we only used a simple MCMC algorithm, but a Metropolis-adjusted 
Langevin algorithm or Hamiltonian Monte-Carlo algorithm may be more appropriate 
to solve this problem.

\subsection*{Acknowledgements}
AA acknowledges EPSRC funding through award EP/N014529/1 via
the EPSRC Centre for Mathematics of Precision Healthcare.

\bibliographystyle{abbrv}
\bibliography{select_mm}
\end{document}